\documentclass[conference]{IEEEtran}
\usepackage{cite}
\usepackage{amsmath,amssymb,amsfonts}
\usepackage{algorithmic}
\usepackage{graphicx}
\usepackage{textcomp}
\usepackage{xcolor}
\def\BibTeX{{\rm B\kern-.05em{\sc i\kern-.025em b}\kern-.08em
    T\kern-.1667em\lower.7ex\hbox{E}\kern-.125emX}}

\usepackage{amsthm}
\newtheorem{theorem}{Theorem}
\newtheorem{assumption}{Assumption}
\newtheorem{remark}{Remark}

\usepackage{balance}

\usepackage{soul}

\begin{document}

\title{Wireless Feedback Control with Variable Packet Length for Industrial IoT
}
\author{Kang Huang, Wanchun Liu$^\dagger$, Yonghui Li, Andrey Savkin, and Branka Vucetic\\
}

\maketitle
\begin{abstract}
\let\thefootnote\relax\footnote{K. Huang, W. Liu, Y. Li and B. Vucetic are with School of Electrical and Information Engineering, The University of Sydney, Australia.
	Emails:	\{kang.huang,\ wanchun.liu,\ yonghui.li,\ branka.vucetic\}@sydney.edu.au. 
A. Savkin is with School of Electrical Engineering and Telecommunications,  University of New South Wales, Australia.
Email: a.savkin@unsw.edu.au.	
(\emph{Wanchun Liu is the corresponding author.})
}
The paper considers a wireless networked control system (WNCS), where a controller sends packets carrying control information to an actuator through a wireless channel to control a physical process for industrial-control applications. In most of the existing work on WNCSs, the packet length for transmission is fixed. However, from the channel-encoding theory, if a message is encoded into a longer codeword, its reliability is improved at the expense of longer delay. Both delay and reliability have great impact on the control performance. Such a fundamental delay-reliability tradeoff has rarely been considered in WNCSs. In this paper, we propose a novel WNCS, where the controller adaptively changes the packet length for control based on the current status of the physical process. We formulate a decision-making problem and find the optimal variable-length packet-transmission policy for minimizing the long-term average cost of the WNCSs. We derive a necessary and sufficient condition on the existence of the optimal policy in terms of the transmission reliabilities with different packet lengths and the control system parameter.
\end{abstract}


\section{Introduction}
The Industrial Internet of Things (IIoT)
can be treated as an extension of the consumer IoT in industrial applications.
One of the most important applications for IIoT is industrial control~\cite{Schulz} with user scenarios ranging from building and process automation to more mission-critical applications, such as factory automation and power system control~\cite{wollschlaeger2017future}. 
%
Wireless networked control systems (WNCSs) are composed of spatially distributed controllers, sensors and actuators communicating through wireless channels, and physical processes to be controlled. Due
to the enhanced flexibility and the reduced deployment and maintenance costs, WNCS is becoming a fundamental infrastructure technology for mission-critical control applications~\cite{ParkSurvey}.
In \cite{schenato2007foundations}, the optimal control policy and the stability condition of a WNCS were investigated.
In \cite{GatsisOppor}, the optimal transmission scheduling of multiple control systems over shared communication resources were studied.
In \cite{KangJIoT}, the uplink and downlink transmission scheduling problem of a WNCS with a half-duplex controller were considered.
In \cite{demirel2017trade}, an event-triggered WNCS was proposed to reduce the communication cost. {In \cite{wncs1} and \cite{wncs2}, WNCSs with low-power and high-performance multi-hop wireless networks were investigated, respectively.}

In most of the existing works on WNCSs~\cite{schenato2007foundations,GatsisOppor,KangJIoT,demirel2017trade,KangTWC,KangICC,LiuJIoT,LiuGC}, the status of the physical process was discretized by periodical sampling, and the transmission of the controller's  packet was ideally assumed to be fixed and equal to sampling period.
From the theory of channel encoding, if a message is encoded into a longer codeword (with a longer packet length), it can be delivered to the receiver with a higher reliability, but introduces a longer transmission delay on the other side. 
This introduces the fundamental tradeoff in transmission delay and reliability~\cite{Polyanskiy}.
In a WNCS,
the transmission of a short control-information-carrying packet results in a frequent but unreliable control; while the transmission of a long control packet leads to a less timely but more reliable control.
Packet-length design for achieving the optimal control-system performance has rarely been considered in the open literature of WNCSs.

{Moreover, a WNCS is a dynamic system, and the state of the physical process under control changes with time. Naturally, different status of the WNCS can require different levels of reliability and delay of control-packet transmission for achieving the optimal control performance. For some status, reliable transmission is more important, which needs a longer control packet, while for some other status, short-delay transmission is more important, which needs a shorter control packet. Variable-length packet transmission has been proposed and investigated in conventional communication systems~\cite{variable1,variable2}, which, however, has not been considered in WNCSs.}

In this paper, we tackle the packet-length design problem in WNCSs. The major contributions are  summarized as below.

{
	1) We propose a novel WNCS, where the controller is able to adaptively change the packet length for control based on the current status of the physical process.
	We formulate the optimal design problem of variable-length packet transmission policy into a semi-Markov decision process (MDP) problem, which minimizes the long-term average cost function of the WNCS.
	The problem can be solved by a modified relative value iteration algorithm. Numerical results show that the proposed WNCS achieves a much better performance than the conventional system.
	
	2) 
	For the variable-length policy, we derived the stability condition of the WNCS, i.e., the necessary and sufficient condition on the existence of an optimal policy that can stabilize the WNCS (i.e., make the cost function bounded) in terms of the transmission reliabilities with different packet lengths and the control system parameter.
	The analysis is not trivial, since neither the optimal policy nor its long-term average cost function of the WNCS has a closed-form formula. We prove the necessity by constructing a tractable virtual policy that can achieve a better performance than the optimal policy. The sufficiency is proved by analyzing fixed-length policies that achieve a worse performance than the optimal variable-length policy.

	3) We investigate the fixed-length packet transmission policies of the WNCS with different packet lengths, and derive a closed-form stability condition in terms of the packet length, the transmission reliability and the control system parameter. Such the result has not been obtained before in the literature, and will provide an important design guideline for the fixed-length policy in WNCS. 
}

%
%


\vspace{-0.2cm}
\section{System Model} \label{model}
Consider a wireless networked control system, where the discrete-time dynamic physical process $\mathbf{x}_t \in \mathbb{R}^n$, $t\in \mathbb{N}$, is measured by the \emph{controller}, which generates and sends control information to the remote \emph{actuator} to control the process. $\mathbb{N}$ is the set of positive integers.
The evolution of the dynamic process is modeled as a linear time-invariant  system~\cite{schenato2007foundations,GatsisOppor,KangJIoT,demirel2017trade}:
\begin{equation} \label{system_model}
\mathbf{x}_{t+1} = \mathbf{A}\mathbf{x}_t + \mathbf{B}\mathbf{u}_t + \mathbf{w}_t,
\end{equation}
where $\mathbf{u}_{t} \in \mathbb{R}^m$ is the actuator's control input, $\mathbf{w}_{t} \in \mathbb{R}^n$ is
the process disturbance modeled as a zero-mean Gaussian white noise with the covariance $\mathbf{R} \in \mathbb{R}^{n \times n}$, and 
$\mathbf{A} \in \mathbb{R}^{n \times n}$ and $\mathbf{B} \in \mathbb{R}^{n \times m}$ are the system transition matrix and the input matrix, respectively. The discrete time slot has the duration of $T_0$~s, which is also the sampling period of the process. For brevity, we only use the discrete time for analysis in the rest of the paper.

\subsection{Controller-Side Operation}
To deliver the control information to the actuator, the controller converts its control signal into a packet by quantization and channel encoding (i.e., error-control coding).
The communication channel for packet transmission is static for low-mobility industrial control applications~\cite{schenato2007foundations,KangJIoT,demirel2017trade}.
We assume that the quantization noise is negligible due to the sufficiently high number of quantization levels, which is commonly considered in the literature~\cite{schenato2007foundations,GatsisOppor,KangJIoT,demirel2017trade,KangTWC,KangICC,LiuJIoT}.
Since a longer channel-coding blocklength leads to a longer packet with a higher reliability~\cite{Polyanskiy}, the packet error probability is a monotonically decreasing function $g(l)$ in terms of the packet length of $l$ time slots, where $l \in \mathbb{N}$.

The transmission of control information introduces delay, so we adopt a predictive control method for delay compensation~\cite{demirel2017trade,KangJIoT}.
To be specific, given current time $t$ and packet length $l$, since the control information is expected to be delivered at time $(t+l-1)$ and there is no control input until then, the controller optimally predicts the process state $\mathbf{x}_{t+l-1}$ as~\cite{KangJIoT}
\begin{equation} \label{state_estimation}
\mathbf{\hat{x}}_{t+l-1|t} = \mathbf{A}^{l-1}\mathbf{x}_t.
\end{equation}
By adopting a linear control law, the control signal that is generated at time $t$ and to be applied at  time $(t+l-1)$ by the actuator is~\cite{KangJIoT,demirel2017trade}
\begin{equation} \label{predictive_control}
\mathbf{\hat{u}}_{t+l-1|t} = \mathbf{K} \mathbf{\hat{x}}_{t+l-1|t} = \mathbf{K}\mathbf{A}^{l-1}\mathbf{x}_t,
\end{equation}
where $\mathbf{K}\in \mathbb{R}^{m\times n}$ is the constant controller gain. 
\begin{assumption}
	\normalfont
The controller gain has the property that~\cite{KangJIoT} 
\begin{equation}\label{a+bk}
\mathbf{A+BK=0}.
\end{equation}
\end{assumption}
This assumes that the control system is one-step controllable\footnote{The multi-step controllable cases can also be handled by the following problem formulation and performance analysis framework, and the stability conditions in Theorems~\ref{theorem_fixed} and~\ref{theorem_var} remain the same.}, i.e., it can be verified by taking \eqref{predictive_control} into \eqref{system_model} that the state vector $\mathbf{x}_t$ can be driven to zero in one time slot in the absence of the process disturbance with the packet length~$l=1$.
	
\subsection{Actuator-Side Operation}
Let $\gamma_t =1$ denote a successful packet detection at time $t$ and $\tilde{l}_t$ denote the length of the successfully received packet.
Thus, $\gamma_t =0$ denotes the packet has not arrived at $t$ or the detection of the arrived packet at $t$ is failed.

{Now, we introduce the the age-of-information (AoI) at the actuator,  $d_t$, which measures the time duration between the generation time of the most recently received control packet and the current time $t$~\cite{kaul2012real,KangTWC}. Then, it is easy to have the updating rule of $d_t$ as
\begin{equation} \label{packet_update}
d_{t+1} = \begin{cases}
\tilde{l}_t,& \gamma_t =1,\\
d_t + 1,& \text{otherwise}.
\end{cases}
\end{equation}}

The actuator adopts a zero-hold strategy: it remains the zero control input until a control packet is successfully detected~\cite{schenato2007foundations,GatsisOppor,KangJIoT}.
Thus, the actuator's control input $\mathbf{u}_t$ in \eqref{system_model} is given as
\begin{equation}\label{u}
\mathbf{u}_{t} = \begin{cases}
\mathbf{\hat{u}}_{t|t-\tilde{l}_t+1}, & \gamma_t =1 \\
\mathbf{0}, & \text{otherwise.}
\end{cases}
\end{equation}

Taking \eqref{u} into \eqref{system_model} and using the property \eqref{a+bk}, the state covariance matrix can be obtained as 
\begin{equation}\label{Pt}
\mathbf{P}_t \triangleq \mathbb{E}\left[\mathbf{x}_t \mathbf{x}^{\top}_t\right]
=
\mathbf{H}(d_t)
\triangleq \sum_{i=0}^{d_t-1}\mathbf{A}^{i}\mathbf{R}(\mathbf{A}^{i})^{\top},
\end{equation}
where $\mathbb{E}[\cdot]$ is the expectation operator, $(\cdot)^\top$ is the operator of matrix transpose.
Therefore, the state covariance matrix $\mathbf{P}_t$ in \eqref{Pt} depends on the AoI status $d_t$.

\section{Control with Variable-Length Packets}
\subsection{Problem Formulation}
The performance of the control system is measured by the quadratic average cost function as~\cite{schenato2007foundations,KangJIoT,demirel2017trade}
\begin{equation} \label{average_cost}
J = \lim_{T \to \infty}\frac{1}{T}\sum_{t=1}^{T}
\mathbb{E}\left[\mathbf{x}_t^{\top}\mathbf{Q}\mathbf{x}_t\right] 
= \lim_{T \to \infty}\frac{1}{T}\sum_{t=1}^{T}\text{Tr}\left(\mathbf{Q}\mathbf{P}_t\right),
\end{equation}
where $\mathbf{Q}$ is a symmetric positive semidefinite weighting matrix, and $\text{Tr}\left(\mathbf{Q}\mathbf{P}_t\right)$ is the cost function at time $t$.

We define the variable-length packet transmission policy for wireless control: the policy $\pi = \{l_{1},l_{2},..,l_{k},...\}$ is the sequence of the packet lengths during the process control, where $l_k\in\mathbb{N}$ and $k$ is the packet index, as illustrated in Fig.~\ref{fig:process}.
Our problem is to find the optimal policy $\pi^*$ that minimizes the infinite-horizon average cost, i.e.,
\begin{equation} \label{problem}
\pi^\star \triangleq  \underset{\pi}{\mathrm{argmin}} \lim_{T \to \infty}\frac{1}{T}\sum_{t=1}^{T}\text{Tr}\left(\mathbf{Q}\mathbf{P}_t\right).
\end{equation}

\begin{figure}[t]
	\centering
	\includegraphics[scale=0.8]{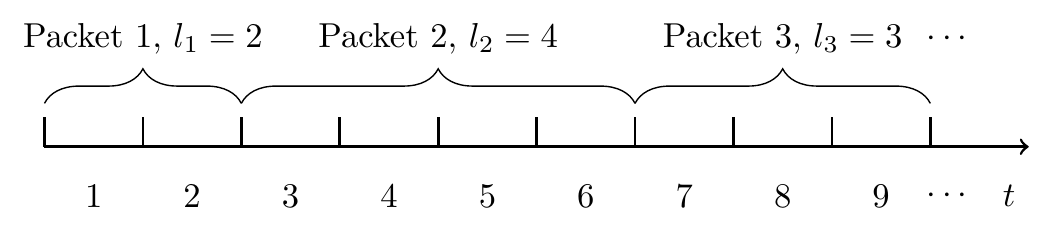}
	\vspace{-0.3cm}
	\caption{Variable-length packet transmission policy.}
	\vspace{-0.5cm}	
	\label{fig:process}
\end{figure}

\subsection{Semi-MDP Solution}\label{sec:semi}
From the definition of $\mathbf{P}_t$ in \eqref{Pt} and the updating rule of the AoI status \eqref{packet_update}, the problem \eqref{problem} can be treated as an adaptive packet-length decision process with two properties: 1) if a decision of packet length $l$ is made at the AoI state $d$, then the each step cost $\text{Tr}(\mathbf{QP}_t)$ depends only on the present AoI state $d$ until the completion of packet transmission;
2) the sum cost incurred until the completion of the packet transmission depends only on the
present AoI state $d$ and the packet length~$l$.

Such an average-cost minimization problem is a typical semi-Markov decision process (semi MDP)~\cite{tijms2003first} modeled as:

1) The state space is defined as $\mathbb{S} \triangleq \{d: d \in \mathbb{N} \}$. The state indicates the AoI at the beginning of a packet transmission.
The state at the beginning of the $k$th packet is denoted as $d_k \in \mathbb{S}$.

2) The action space is defined as $\mathbb{A} \triangleq \{l: l \in \mathbb{N} \}$.
The action $l_k \triangleq \pi(d_k)$ represents the length of the $k$th packet, where, with a slight abuse of notation, $\pi(d_k)$ is the stationary policy function in terms of the current state $d_k$.

3) The state-transition probability $P(d'|d,l)$ characterizes the probability that state transits from $d$ at the beginning of the current packet to $d'$ at the beginning of the next packet with the current action of $l$. As the transition is time-homogeneous, we drop the packet index $k$ here. From the state-updating rule \eqref{packet_update} and the packet error probability function $g(l)$, the state transition probability is:
\begin{equation} \label{transition_function}
P(d'\vert d, l) = \begin{cases}
g(l) &\text{ if } d' = d+l\\
1-g(l) &\text{ if } d' = l.
\end{cases}
\end{equation}

4) The duration time $\delta(d,l)$ characterizes the expected time until the next transmission decision if the action $l$ is chosen at the current state $d$. It is clear that the duration time is determined by the decided packet length in our scenario, i.e., 
\begin{equation} \label{duration}
\delta(d,l) = l.
\end{equation} 

5) The one-stage cost of the semi MDP problem is the sum cost during the current packet transmission, which is given as:
\begin{equation}
c(d,l) \triangleq \sum_{i=d}^{d+l-1} \text{Tr}\left(\mathbf{Q}\mathbf{H}(i)\right).
\end{equation}

From the semi-MDP formulation, the average cost $J$ in \eqref{problem} can be rewritten as:
\begin{equation}\label{problem2}
J = \frac{\sum_{d \in \mathbb{S}}c(d,\pi(d))\phi_{\pi}(d)}{\sum_{d \in \mathbb{S}}\delta(d,\pi(d))\phi_{\pi}(d)},
\end{equation}
where $\phi_{\pi}(d)$ denotes the stationary probability of state $d\in\mathbb{S}$ under policy $\pi$.

Therefore, the optimal policy of problem \eqref{problem}, $\pi^\star(\cdot)$, can be obtained by solving the above semi-MDP problem with the target function \eqref{problem2}.
By using the classical data-transformation method, the semi-MDP problem can be transformed as an MDP problem, and thus can be solved effectively by the classical relative value iteration method~\cite{tijms2003first}.

\subsection{Practical Implementation Issues of Variable-Length Policy}
{Since the control packet length of a WNCS changes with time, each packet header should include the information of the packet length. Thus, comparing with a fixed-length policy, the variable-length policy requires a bit higher communication overhead in practice. Moreover, when considering multiple WNCSs sharing the same wireless resources, it is not applicable to consider a time-division multiple access like medium-access control (MAC) protocol for multi-WNCS scheduling as each WNCS requires dynamic time slot length for transmissions. Therefore, different WNCSs need to be allocated to different frequency channels/sub-carriers when applying the variable-length policy.
Commonly considered error-control codes, including polar codes and turbo codes, can be used for variable-length encoding.}

\section{Stability Condition of the Control System}
If the packet transmissions are very unreliable with different packet lengths, the average cost $J$ in \eqref{problem2} might be unbounded no matter what packet-transmission policy we choose and the semi-MDP problem discussed above might not have a feasible solution, i.e., the control system is \emph{unstable}.

Now, we study the stability condition of the control system by investigating the fixed-length and variable-length packet transmission policies in the sequel, where the latter is a more general case of the former.
\subsection{Fixed-Length Packet Transmission Policy}

\begin{theorem}[Fixed-length scenario] \label{theorem_fixed}
\normalfont	
Consider the control system described by \eqref{system_model}-\eqref{u}. 
Let $(\mathbf{A},\mathbf{\sqrt{R}})$ be controllable and let $(\mathbf{A},\mathbf{\sqrt{Q}})$ be observable.\footnote{$\sqrt{\mathbf{R}}$ and $\sqrt{\mathbf{Q}}$ are the square roots of the positive definite matrices
$\mathbf{R}$ and $\mathbf{Q}$, respectively.	
$(\mathbf{A},\mathbf{\sqrt{R}})$ is controllable and $(\mathbf{A},\mathbf{\sqrt{Q}})$ is observable if  $\left[\sqrt{\mathbf{R}},\mathbf{A}\sqrt{\mathbf{R}},\cdots,\mathbf{A}^n\sqrt{\mathbf{R}}\right]$
and
$\left[\sqrt{\mathbf{Q}}^{\top},\mathbf{A}^{\top}\sqrt{\mathbf{Q}}^{\top},\cdots,(\mathbf{A}^n)^{\top}\sqrt{\mathbf{Q}}^{\top}\right]$
are of full rank.}
Assuming that the packet length is $l$ and fixed during the process control, the dynamic process can be stabilized iff
\begin{equation}\label{fixed}
g(l) \rho^{2l}(\mathbf{A}) <1,
\end{equation}
where $\rho(\mathbf{A})$ is the spectral radius of the matrix $\mathbf{A}$.
\end{theorem}

\begin{proof}
Consider the policy with fixed packet length $l_0 \in \mathbb{N}$, i.e., $\pi(d) = l_0, \forall d \in \mathbb{S}$. From \eqref{packet_update}, it is easy to see that the state space in Section~\ref{sec:semi} is degraded  into $\mathbb{S}=\{l_0,2l_0,3l_0,\cdots\}$.
Since the packet error probability is fixed, it can be proved that the process of the AoI, $\{d_k\}$, has the stationary distribution as
\begin{equation}\label{geometric}
\phi_{\pi}\big(il_0\big) = (1-g(l_0))g(l_0)^{i-1}, i = 1,2,\cdots
\end{equation}

From \eqref{duration} and \eqref{geometric}, it is clear that the denominator of \eqref{problem2} is bounded. Thus, the average cost $J$ is bounded iff the numerator of \eqref{problem2} is.
Using the inequalities below~\cite[Lemma 1]{KangJIoT}, 
\begin{equation}\label{inequal}
\text{Tr}\left(\mathbf{Q}\mathbf{H}(d)\right)\leq c(d,l)
\leq l \text{Tr}\left(\mathbf{Q}\mathbf{H}(d+l)\right),
\end{equation}
we can obtain
\begin{align} \label{J_upbound}
&\sum_{i = 1}^{\infty}c(il_0,l_0) \phi_{\pi}\big(il_0\big)
< 
\frac{l_0(1-g(l_0))}{g^2(l_0)}\sum_{i=2}^{\infty}g^{i}(l_0)\text{Tr}\left(\mathbf{Q}\mathbf{H}(il_0)\right),\\
 \label{J_lowerbound}
&\sum_{i = 1}^{\infty}c(il_0,l_0) \phi_{\pi}\big(il_0\big) \geq (1-g(l_0))\sum_{i=1}^{\infty}g^{i}(l_0)\text{Tr}\left(\mathbf{Q}\mathbf{H}(il_0)\right).
\end{align}


From~\cite{schenato2007foundations}, if $q>0$ and  $(\mathbf{A},\mathbf{\sqrt{R}})$ and $(\mathbf{A},\mathbf{\sqrt{Q}})$ are controllable and observable, respectively, the following property holds:
\begin{equation} \label{ref}
\sum_{i=1}^{\infty} q^i \text{Tr}\left(\mathbf{Q}\mathbf{H}(i)\right)< \infty \text{ iff }
q \rho^2(\mathbf{A})< 1.
\end{equation}
Applying \eqref{ref} to \eqref{J_upbound} and \eqref{J_lowerbound}, it can be obtained that the average cost $J$ is bounded iff $\left(g(l_0)\right)^{1/{l_0}}\rho^2(\mathbf{A})<1$, which completes the proof of Theorem~\ref{theorem_fixed}.
\end{proof}
\begin{remark}
	\normalfont
Theorem~\ref{theorem_fixed} says that the stability condition under the fixed-length policy depends on the packet error probability, the length of the packet and the control system parameter.
The process~\eqref{system_model} can be stabilized, if the packet length $l$ is properly chosen such that both the packet error probability $g(l)$ and the $l$th power of $\rho^2(\mathbf{A})$ are small.
\end{remark}

\subsection{Variable-Length Packet Transmission Policy}
\begin{theorem}[Variable-length scenario]\label{theorem_var}
\normalfont	
Consider the same system and conditions as defined in Theorem~\ref{theorem_fixed}.
There exists a stationary and deterministic variable-length packet transmission policy that can stabilize the dynamic process iff 
\begin{equation}\label{var}
\min_{l\in\mathbb{N}} g(l) \rho^{2l}(\mathbf{A}) <1. 
\end{equation}
\end{theorem}
\begin{proof}
The sufficiency is easy to prove based on Theorem~\ref{theorem_fixed} as the optimal variable-length policy results in an average cost no higher than that of a fixed-length policy. 

We use a constructive method to prove the necessity.
First, we consider a virtual updating rule below to replace \eqref{packet_update} 
\begin{equation} \label{packet_update2}
d_{t+1} = \begin{cases}
1,& \gamma_t =1\\
d_t + 1,& \text{otherwise}.
\end{cases}
\end{equation}
It is clear that \eqref{packet_update2} is no larger than \eqref{packet_update} for all $t\in \mathbb{N}$, and hence the average cost of the optimal packet-transmission policy by using the updating rule \eqref{packet_update2} is no higher than that of \eqref{packet_update}. Then, we will show that \eqref{var} holds if the average cost induced by a policy is bounded under the condition of \eqref{packet_update2}.

Consider a general policy $\pi'(\cdot)$ that has $\pi'(1)= l'_1$. The state space can be rewritten as $\mathbb{S}=\{1,1+l'_1,1+l'_1+l'_2,1+l'_1+l'_2+l'_3,\cdots\}$, where $l'_k = \pi'(1+\sum_{i=1}^{k-1}l'_i),\forall k\in \{2,3,\cdots\}$.
The average cost function in \eqref{problem2} with packet-transmission policy $\pi'(\cdot)$ can be rewritten as
\begin{equation}\label{problem3}
J= \frac{\sum_{i=1}^{\infty}c'\left(\sum_{j=1}^{i}l'_j\right)\phi'_{\pi'}\left(\sum_{j=1}^{i}l'_j\right)}{\sum_{i=1}^{\infty} \left(\sum_{j=1}^{i}l'_j\right)\phi'_{\pi'}\left(\sum_{j=1}^{i}l'_j\right)},
\end{equation}
where 
\begin{equation} \label{ineq_c}
c'(l) = c(1,l) \geq \text{Tr}\left(\mathbf{Q}\mathbf{H}(l)\right),
\end{equation}
\begin{equation} \label{ineq_phi}
\phi'_{\pi'}\left(\!\sum_{j=1}^{i}l'_j\!\right)\! =\! \prod_{j=1}^{i-1}g(l'_j)\left(1-g(l'_i)\right)
\!\geq\! \left(1-g(1)\right) \!\prod_{j=1}^{i}g(l'_j).
\end{equation}
Since the function $c'(l)$ grows exponentially fast with $l$, $J$ in \eqref{problem3} is bounded iff the numerator is.
Then, by using the inequalities \eqref{ineq_c} and \eqref{ineq_phi}, the numerator of \eqref{problem3} is lower bounded by 
\begin{equation}\label{lb}
\left(1-g(1)\right) \sum_{i=1}^{\infty} \prod_{j=1}^{i}g(l'_j)\text{Tr}(\mathbf{Q}\mathbf{H}\left(\sum_{j=1}^{i}l'_j\right)).
\end{equation}
From \eqref{ref}, it can be proved that 
$\text{Tr}\left(\mathbf{Q}\mathbf{H}(i)\right)$ grows up as fast as $\rho^{2i}(\mathbf{A})$ when~$i \rightarrow \infty$.
Thus, $\text{Tr}\left(\mathbf{Q}\mathbf{H}\left(\sum_{j=1}^{i}l'_j\right)\right)$ can be approximated by $\eta \rho^{2 \left(\sum_{j=1}^{i}l'_j\right)}(\mathbf{A})$ when $i$ is large, where $\eta>0$.
Thus, if \eqref{lb} is bounded, 
$\min_{j\in \mathbb{N}}g(l'_j)\rho^{2l'_j}(\mathbf{A})<1$ holds, which completes the proof of Theorem~\ref{theorem_var}.
\end{proof}
\begin{remark}
	\normalfont
Theorem~\ref{theorem_var} shows that the stability condition of variable-length packet transmission policy is looser than that of a fixed-length policy in Theorem~\ref{theorem_fixed}. The stability condition depends on the function of packet error probability $g(l)$ and also the system parameter $\mathbf{A}$.
\end{remark}

\section{Numerical Results}
In this section, we numerically evaluate the optimal variable-length packet transmission policy in the WNCS and compare it with the fixed-length packet transmission policies.
In order to find the optimal policy, we need to solve the semi-MDP problem with finite state and action spaces. Thus, the infinite state space $\mathbb{S}$ is truncated as $\mathbb{S} = \{1,\cdots,N\}$. 
The action space is $\mathbb{A} = \{1,\cdots,M\}$.
The function of packet error probability in terms of the packet length is approximated by an exponential function as $g(l) = 0.8 \times  0.5^{l-1}$~\cite{tripathi2003reliability,ceran2018average,KangTWC}.
{Unless otherwise stated, we set $N=70$ and $M=5$ for solving the variable-length policy, and consider a scalar system~\cite{LiuJIoT,scalar1,scalar3}, where $\mathbf{A} = 1.2$, $\mathbf{B} = 1$, $\mathbf{R} = 1$, $\mathbf{Q} = 1$~\cite{scalar1}, and thus $\rho(\mathbf{A}) = 1.2$ and $\mathbf{K} = -1.2$.}

\begin{figure}[t]
	\centering
	\includegraphics[scale=0.5]{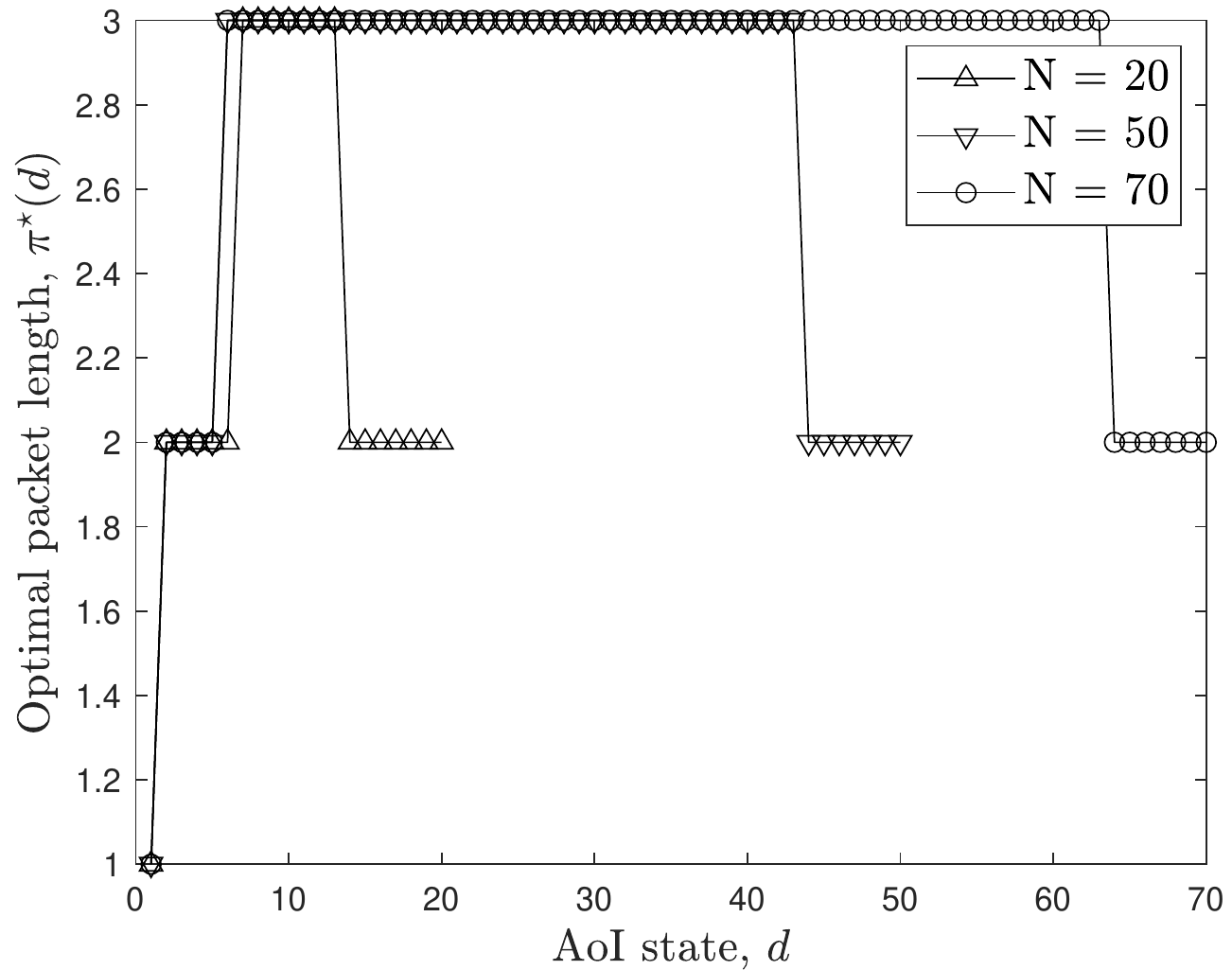}
	\vspace{-0.3cm}
	\caption{The optimal variable-length packet transmission policies within different truncated state spaces.}
	\vspace{-0.3cm}	
	\label{fig:policy}
\end{figure}

Fig.~\ref{fig:policy} shows the optimal packet-transmission policy of the semi-MDP problem with different truncated state-space cardinality $N$. 
It is interesting to see that when the AoI state is small ($d\leq7$), the optimal packet length in different truncated state spaces are almost the same, and the optimal packet length increases with the increasing AoI state.
Also, we see that when the state space is large (i.e., $N=70$), which is closer to the ideal infinite state-space case, the optimal packet length tends to be invariant when $d>7$. Thus, it is reasonable to infer that the optimal policy with the infinite state space has the property that $\pi^\star(d)=3$ when $d> 7$.
The structure of the optimal policy shows that 
when the current system AoI is pretty good, it is wise to take the risk of a transmission with a lower reliability to achieve a good control quality, as the control quality will not be too bad even if the transmission~fails.

\begin{figure}[t]
	\centering
	\includegraphics[scale=0.5]{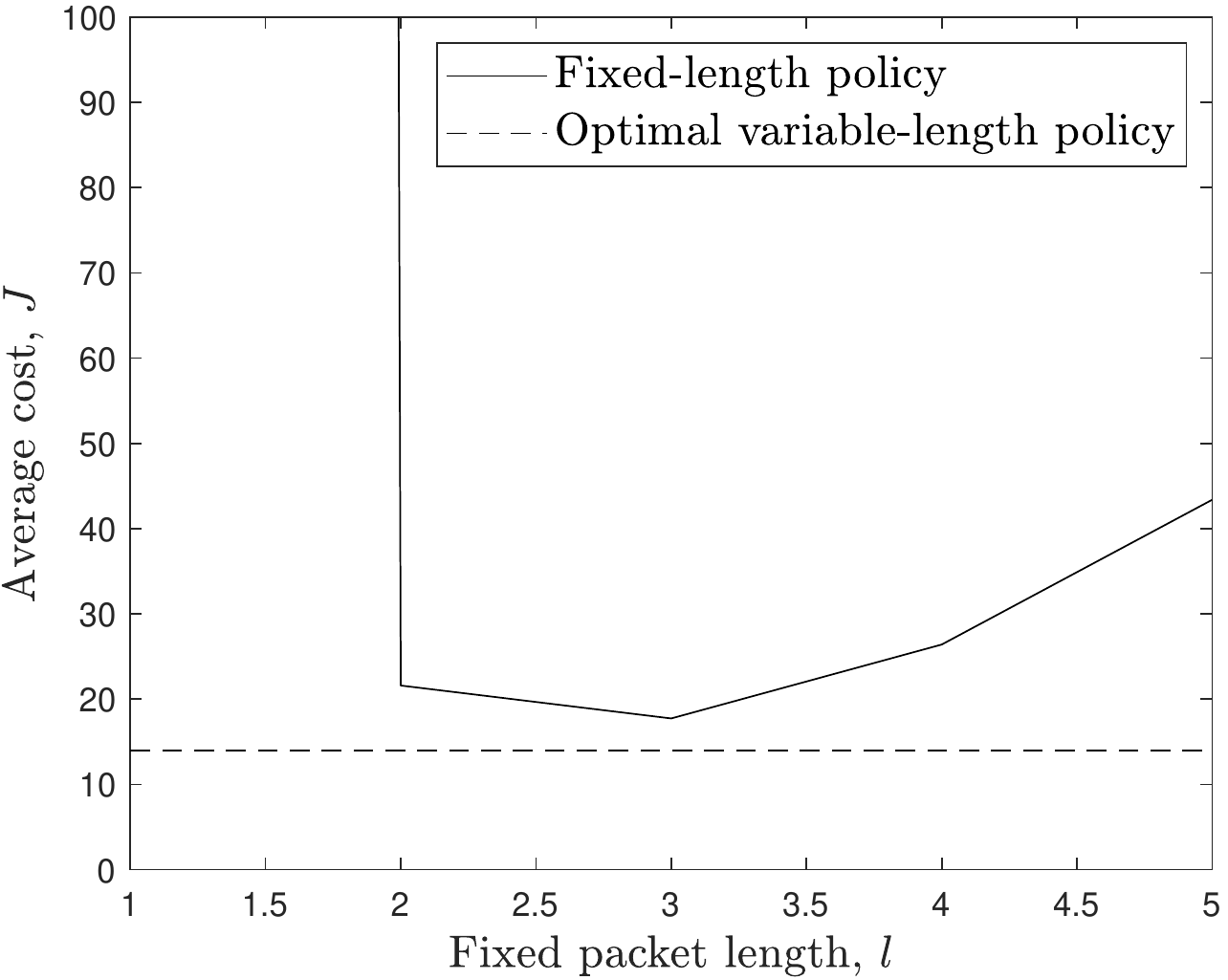}
	\vspace{-0.3cm}	
	\caption{The average cost of fixed-length packet-transmission policy versus the packet length, and the average costs of optimal variable-length policy.}
	\vspace{-0.5cm}	
	\label{fig:comparison}
\end{figure}

Fig. \ref{fig:comparison} plots the average costs of the fixed-length packet transmission policies with different packet lengths and optimal variable-length policy based on \eqref{average_cost} with $T=50,000$. 
From Theorem~\ref{theorem_fixed}, it can be verified that when the fixed packet length $l=1$, which is the conventional transmission policy in most of the existing work~\cite{schenato2007foundations,GatsisOppor,KangJIoT,demirel2017trade,KangTWC,KangICC,LiuJIoT}, the control system is unstable.
From Fig.~\ref{fig:comparison}, the system can be stabilized with longer transmission packets, and the average cost is minimized when the fixed packet length is $3$. 
This optimal fixed-length policy is largely in agreement with the optimal variable-length policy illustrated in Fig.~\ref{fig:policy}, where the optimal packet length is $3$ for most of the states in the state space. 
Also, we see that the optimal variable-length policy gives a $22\%$ average cost reduction of the optimal fixed-length policy, which shows the importance of adaptive packet-transmission design in WNCSs. 

\section{Conclusions}
In this paper, we have proposed and optimized the variable-length packet transmission policy. We have also derived the control-system stability conditions for both the fixed-length and variable-length policies. Our numerical results have demonstrated the superior of the proposed variable-length packet transmission method in wireless control systems.
%
\balance


\end{document}